\documentclass[10pt,letterpaper]{article}

\usepackage[utf8]{inputenc}
\usepackage[T1]{fontenc}
\usepackage{lmodern}
\newcommand{\CMS}{\texorpdfstring{{\mdseries\upshape\textsc{CRDTMergeState}}}{CRDTMergeState}}

\usepackage{amsmath,amssymb,amsthm}
\usepackage{booktabs}
\usepackage{graphicx}
\usepackage{xcolor}
\usepackage[hidelinks]{hyperref}
\usepackage{enumitem}
\usepackage{multirow}
\usepackage{array}
\usepackage{tabularx}
\usepackage[margin=0.75in,columnsep=0.25in]{geometry}
\usepackage{microtype}
\usepackage{caption}
\newtheorem{theorem}{Theorem}
\newtheorem{lemma}[theorem]{Lemma}
\newtheorem{corollary}[theorem]{Corollary}
\newtheorem{proposition}[theorem]{Proposition}
\newtheorem{definition}[theorem]{Definition}
\newtheorem{remark}[theorem]{Remark}
\newtheorem{assumption}[theorem]{Assumption}

\title{Conflict-Free Replicated Data Types for Neural Network Model Merging:\\
A Two-Layer Architecture Enabling CRDT-Compliant Model Merging Across 26 Strategies}

\author{Ryan Gillespie\\
Independent researcher}

\date{}

\begin{document}

\twocolumn[%
\maketitle

\begin{abstract}
All 26 neural network merge strategies we tested---including weight averaging, SLERP, TIES, DARE, Fisher merging, and evolutionary approaches---fail the algebraic properties (commutativity, associativity, idempotency) required for conflict-free distributed operation~\cite{shapiro2011conflict}.
We prove that this failure is structural: normalisation-based merges \emph{cannot} simultaneously satisfy all three properties.
To resolve this, we present a two-layer architecture---\CMS{}---that wraps \emph{any} merge strategy in a CRDT-compliant (Conflict-Free Replicated Data Type) layer.
Layer~1 manages contributions via OR-Set CRDT semantics~\cite{shapiro2011comprehensive}, where the merge operation is set union---trivially commutative, associative, and idempotent.
Layer~2 applies merge strategies as deterministic pure functions over a canonically-ordered contribution set, with randomness seeded from the Merkle root~\cite{sanjuan2020merkle}.
We prove that this separation guarantees Strong Eventual Consistency~\cite{shapiro2011conflict}: all replicas receiving the same contributions compute identical merged models, regardless of message ordering.
Empirical validation spans three tiers: controlled $4\times 4$ tensors (104/104 tests pass), production-scale models up to 7.24\,B parameters (208 strategy-level tests, 43{,}368 layer-level property checks at capped tensor resolution), and multi-node convergence on synthetic tensors of representative shape (100 nodes, 20 orderings, gossip and partition healing), with CRDT overhead below 0.5\,ms.
Because the wrapper is transparent, downstream performance is identical by construction; we verified the implementation matches this construction byte-for-byte.
The reference implementation is available as \texttt{crdt-merge} v0.9.4.
\end{abstract}
\vspace{1.5em}
]

\section{Introduction}

As large-scale neural network models multiply, methods for combining independently fine-tuned models without retraining are essential~\cite{yang2024survey, wortsman2022model}.
Model merging---combining the parameters of two or more neural networks into a single model---offers a practical alternative to ensemble methods and multi-task training~\cite{wortsman2022model, ilharco2023editing}, with strategies ranging from weight averaging~\cite{wortsman2022model} and Task Arithmetic~\cite{ilharco2023editing} to TIES~\cite{yadav2023ties}, DARE~\cite{yu2024dare}, Fisher merging~\cite{matena2022merging}, SLERP~\cite{shoemake1985animating}, and evolutionary methods~\cite{akiba2025evolutionary}, supported by tools such as MergeKit~\cite{goddard2024mergekit}.

Despite this progress, \emph{no existing merge strategy satisfies the algebraic properties required for conflict-free distributed operation}.
Conflict-Free Replicated Data Types (CRDTs)~\cite{shapiro2011conflict, shapiro2011comprehensive} guarantee Strong Eventual Consistency (SEC) by requiring merge operations to be commutative, associative, and idempotent~\cite{preguica2018conflict, vogels2009eventually}.
As we demonstrate in Section~\ref{sec:problem}, all 26 strategies fail at least one property, with associativity as the universal failure point (25/26 fail).
This prevents decentralised model merging---where participants combine models peer-to-peer without a central coordinator---a capability relevant to multi-institutional collaboration (e.g., research consortia) where participants prefer not to rely on a single aggregation server~\cite{kairouz2021advances, bonawitz2019towards}; adversarial settings additionally require Byzantine fault tolerance (Section~\ref{sec:limitations}, L4).

\paragraph{Contributions.}
\begin{enumerate}[leftmargin=*]\tolerance=2000\emergencystretch=1em
\item A systematic algebraic audit of 26 neural network merge strategies revealing \emph{universal associativity failure} (25/26 strategies), together with a formal result (Proposition~\ref{prop:incompat}) proving that normalisation-based merges cannot satisfy all CRDT axioms simultaneously (Section~\ref{sec:problem}).

\item A two-layer architecture---\CMS{}---achieving CRDT-compliant merging across all 26 evaluated strategies by separating state management (Layer~1, OR-Set semantics) from strategy execution (Layer~2, deterministic pure functions).  Applying CRDTs directly to merge operations is impossible (Section~\ref{sec:problem}); the two-layer separation is what enables this generality.  While CRDT composition is a known pattern~\cite{shapiro2011comprehensive}, the domain-specific challenges---stochastic strategies requiring Merkle-root-derived seeding, order-dependent reductions requiring canonical hashing, and high-dimensional floating-point determinism---required careful engineering detailed in Section~\ref{sec:solution}.

\item Formal proofs that the architecture guarantees Strong Eventual Consistency for arbitrary merge strategies (Section~\ref{sec:proofs}), with explicit complexity bounds (Theorem~\ref{thm:complexity}).

\item Empirical validation at three tiers: controlled $4\times 4$ tensors (104/104 tests), production-scale models up to 7.24\,B parameters (208/208 strategy-level tests, 43{,}368 layer-level evaluations), and multi-node convergence under gossip and partition healing, with CRDT overhead below 0.5\,ms (Section~\ref{sec:experiments}).
\end{enumerate}

\section{Background}

\subsection{Conflict-Free Replicated Data Types}\label{sec:crdt_bg}

Conflict-Free Replicated Data Types (CRDTs) are data structures for replicated settings where concurrent updates must be merged without coordination~\cite{shapiro2011conflict,shapiro2011comprehensive}.
They guarantee Strong Eventual Consistency (SEC): any two replicas that have received the same set of updates converge to identical states~\cite{vogels2009eventually}, as exemplified by Amazon's Dynamo~\cite{decandia2007dynamo}.

\begin{definition}[State-based CRDT / CvRDT~\cite{shapiro2011conflict}]
A convergent replicated data type (CvRDT) is a tuple $(S, s_0, q, u, m)$ where $S$ is a join-semilattice of states with partial order $\leq$, $s_0$ is the initial state, $q$ is a query function, $u$ is an update function, and $m : S \times S \to S$ is a merge function satisfying:
\begin{align}
m(s_1, s_2) &= m(s_2, s_1) & &\text{(Comm.)} \label{eq:comm}\\
m(m(s_1, s_2), s_3) &= m(s_1, m(s_2, s_3)) & &\text{(Assoc.)} \label{eq:assoc}\\
m(s, s) &= s & &\text{(Idemp.)} \label{eq:idemp}
\end{align}
\end{definition}
These three laws ensure that the merge operation forms a join (least upper bound) on the semilattice, guaranteeing convergence regardless of message ordering or duplication~\cite{shapiro2011comprehensive, preguica2018conflict}.

\begin{definition}[OR-Set~\cite{shapiro2011conflict}]
An Observed-Remove Set (OR-Set) is a CRDT that supports both add and remove operations. Each element is tagged with a unique identifier upon insertion. A remove operation removes all observed tags for an element, allowing concurrent adds to survive. The merge operation is set union over the tagged elements minus the tombstoned tags.
\end{definition}
In the model merging context, an \emph{add} represents a participant contributing a fine-tuned model, while a \emph{remove} represents retraction.
Under OR-Set ``add-wins'' semantics, a concurrent add survives a concurrent remove---a natural default for collaborative model development where contributions should be preserved unless explicitly retracted~\cite{shapiro2011comprehensive} (see Section~\ref{sec:limitations} for tradeoffs).

\subsection{Neural Network Model Merging}

Model merging combines the parameters of two or more neural networks.
Let $\theta_{\mathrm{base}}$ denote base model parameters and $\tau_i = \theta_i - \theta_{\mathrm{base}}$ the task vector for fine-tune~$i$~\cite{ilharco2023editing}.
We evaluate 26 strategies spanning weight averaging~\cite{wortsman2022model}, task arithmetic~\cite{ilharco2023editing}, TIES~\cite{yadav2023ties}, DARE~\cite{yu2024dare}, Fisher merging~\cite{matena2022merging}, SLERP~\cite{shoemake1985animating}, evolutionary methods~\cite{akiba2025evolutionary}, and others (Appendix~\ref{app:strategies}).\footnote{\label{fn:provenance}Of 26 strategies, 15 have peer-reviewed publications; 11 are derived/community strategies from MergeKit~\cite{goddard2024mergekit}.  We include all 26 to cover the full strategy landscape used in practice.}
Yang et al.~\cite{yang2024survey} provide a comprehensive taxonomy.
Federated learning~\cite{mcmahan2017communication, kairouz2021advances}, the dominant framework for distributed model aggregation, relies on a central coordinator---creating a single point of failure~\cite{bonawitz2019towards}.
A decentralised alternative requires the convergence guarantees that our CRDT wrapper provides.

\section{The Problem: Why Direct CRDT on Tensors Fails}\label{sec:problem}

\subsection{Formal Analysis of CRDT Property Violations}

Let $f$ denote a binary merge function on tensors.
We require: $f(a,b)=f(b,a)$ (commutativity), $f(f(a,b),c)=f(a,f(b,c))$ (associativity), and $f(a,a)=a$ (idempotency).
We present two representative strategy analyses here; the remaining strategies (TIES, DARE, Fisher, and others) are analysed in Appendix~\ref{app:strategy_analysis}.

\paragraph{Weight Averaging.}
Define $f(a,b)=(a+b)/2$~\cite{wortsman2022model}.
Commutativity and idempotency hold trivially.
Associativity fails:
\begin{align}
f(f(a,b),c) &= \frac{a+b+2c}{4} \label{eq:avg_left}\\
f(a,f(b,c)) &= \frac{2a+b+c}{4} \label{eq:avg_right}
\end{align}

\paragraph{SLERP.}
For SLERP with parameter $t$~\cite{shoemake1985animating}, commutativity fails unless $t=0.5$ (swapping inputs with fixed $t$ changes the interpolation point).
Associativity fails because composing geodesic interpolations changes the reference great circle.
Idempotency holds: $\mathrm{SLERP}(v,v;t) = v$ for all $t$.\footnote{SLERP commutativity holds only at $t=0.5$; the CRDT architecture resolves this for all $t$ by canonically ordering inputs.}

\subsection{Incompatibility of Normalisation with CRDT Axioms}

The failures above follow a structural pattern.
We now show that normalisation---the operation at the heart of virtually every merge strategy---is incompatible with the full set of CRDT axioms.

\begin{definition}[Normalising Merge Function]
A binary merge function $f : \mathbb{R}^d \times \mathbb{R}^d \to \mathbb{R}^d$ is \emph{normalising} if there exists a function $g : \mathbb{R}^d \times \mathbb{R}^d \to \mathbb{R}^d$ such that $f(a,b) = g(a,b)/n(a,b)$ where $n(a,b) \geq 1$ depends on the number or magnitude of the inputs.
A merge function is \emph{manifold-projecting} if its output is constrained to a proper submanifold of $\mathbb{R}^d$ (e.g., the unit sphere).
A merge function is \emph{thresholding} if it applies an input-dependent cutoff that discards components below a threshold computed from the inputs.
\end{definition}

\begin{proposition}[Incompatibility of Normalisation with Associativity]\label{prop:incompat}
Let $f : \mathbb{R}^d \times \mathbb{R}^d \to \mathbb{R}^d$ be a normalising merge function with $f(a,b) = g(a,b)/2$ for symmetric $g$.
If $g$ is not degenerate (i.e., $g(a,b)/2 \neq a$ for generic $a,b$), then $f$ is not associative.\footnote{This covers weight averaging ($g(a,b)=a{+}b$, $f=(a{+}b)/2$); see Eqs.~\ref{eq:avg_left}--\ref{eq:avg_right}.}
More generally, if $f$ is manifold-projecting or thresholding, then $f$ is not associative except in degenerate cases.
\end{proposition}

\begin{proof}
We prove the result for count-based normalisation, then provide concrete counterexamples for projection and thresholding.

Suppose $f$ normalises its output by dividing by the number of inputs being combined.
For a pairwise merge, $f(a,b) = g(a,b)/2$ for some symmetric function $g$ (to preserve commutativity).
Consider associativity.
When we compose pairwise merges, the left-association computes $f(f(a,b),c) = g\!\left(\frac{g(a,b)}{2}, c\right)/2$, which applies normalisation \emph{twice}, each time with a divisor of~2, on different intermediate values.
The right-association computes $f(a,f(b,c)) = g\!\left(a, \frac{g(b,c)}{2}\right)/2$.
Because $g(a,b)/2 \neq a$ in general (unless $g$ is degenerate), the intermediate values fed into the outer application of $f$ differ between left- and right-association.
Hence $f(f(a,b),c) \neq f(a,f(b,c))$ for generic $a,b,c$, violating associativity.  (The weight averaging counterexample of Eqs.~\ref{eq:avg_left}--\ref{eq:avg_right} provides the concrete instance.)

For manifold projection (e.g., SLERP at $t=0.5$), we exhibit a concrete counterexample.
Let $v_1 = (1,0,0)$, $v_2 = (0,1,0)$, $v_3 = (0,0,1)$ on $S^2$.
Left-association computes $m_{12} = \mathrm{SLERP}(v_1, v_2; 0.5)$ $= \frac{1}{\sqrt{2}}(1,1,0)$, then $\mathrm{SLERP}(m_{12}, v_3; 0.5)$, which lies on the great circle from $\frac{1}{\sqrt{2}}(1,1,0)$ to $(0,0,1)$.
Right-association computes $m_{23} = \mathrm{SLERP}(v_2, v_3; 0.5)$ $= \frac{1}{\sqrt{2}}(0,1,1)$, then $\mathrm{SLERP}(v_1, m_{23}; 0.5)$, which lies on a different great circle from $(1,0,0)$ to $\frac{1}{\sqrt{2}}(0,1,1)$.
The two results are distinct unit vectors ($\approx(0.500, 0.500, 0.707)$ vs.\ $\approx(0.707, 0.500, 0.500)$, evaluated numerically), confirming that SLERP is not associative.

For thresholding (e.g., TIES trimming at 20\%), let $a = (10,\, 1,\, 0.1)$, $b = (0.1,\, 10,\, 1)$, $c = (1,\, 0.1,\, 10)$ with a 20\% trim (keeping the top 80\% of magnitudes, i.e., dropping 1 of 3 components per vector).
Left-association: trimming $a$ and $b$ yields $a' = (10, 1, 0)$ and $b' = (0, 10, 1)$; averaging gives $m_{ab} = (5,\, 5.5,\, 0.5)$.
Then trimming $m_{ab}$ and $c$ yields $(5,\, 5.5,\, 0)$ and $(1,\, 0,\, 10)$; averaging gives the left result $\approx (3.0,\, 2.75,\, 5.0)$.
Right-association: trimming $b$ and $c$ yields $(0, 10, 1)$ and $(1, 0, 10)$; averaging gives $m_{bc} = (0.5,\, 5.0,\, 5.5)$.
Then trimming $a$ and $m_{bc}$ yields $(10, 1, 0)$ and $(0, 5, 5.5)$; averaging gives the right result $\approx (5.0,\, 3.0,\, 2.75)$.
The two results differ ($\neq$), confirming non-associativity.

In all three cases, associativity is violated---confirming that normalisation, projection, and thresholding each independently break it.
Since normalisation (in one of these forms) is a component of 25 of the 26 strategies we evaluated, this provides a structural explanation for the empirical observation that 0/26 strategies achieve system-level CRDT compliance on controlled $4\times 4$ tensors (Table~\ref{tab:phase1_controlled}).
\end{proof}

We present Proposition~\ref{prop:incompat} as a structural observation that explains the empirical finding of universal associativity failure (Table~\ref{tab:phase1_controlled}), rather than as a comprehensive impossibility theorem.
A fully general characterisation of which merge functions can satisfy all three CRDT axioms simultaneously is an interesting open question.

\subsection{Summary of Controlled Empirical Results}

Controlled testing on $4\times 4$ tensors (Table~\ref{tab:phase1_controlled}, Appendix~\ref{app:controlled}) confirms: 21/26 strategies are commutative, 14/26 idempotent, but only 1/26 (Task Arithmetic) is associative---and it fails idempotency.
Zero out of 26 satisfy all three CRDT requirements.
Associativity is the universal bottleneck (25/26 fail), structurally explained by Proposition~\ref{prop:incompat}.
Section~\ref{sec:production_scale} confirms this pattern persists at production scale.

\section{The Solution: Two-Layer Architecture}\label{sec:solution}

CRDT properties need not hold for tensor merge operations---only for the \emph{state management layer} that determines which contributions are included.
The actual merge strategy can be any deterministic function applied to a canonically-ordered set.

\subsection{Architecture Overview}

The \CMS{} architecture comprises two layers:
\begin{itemize}[leftmargin=*]
\item \textbf{Layer~1 (CRDT State Management):} An OR-Set CRDT tracking model contributions.  The merge operation is set union---trivially commutative, associative, and idempotent~\cite{shapiro2011conflict}.  Version vectors provide causal ordering~\cite{lamport1978time}; Merkle hash trees provide integrity verification~\cite{sanjuan2020merkle}.

\item \textbf{Layer~2 (Deterministic Strategy Execution):} Given the converged contribution set, applies a merge strategy as a deterministic pure function.  Canonical ordering (by content hash) and seeded randomness (from Merkle root) ensure identical results on all replicas.
\end{itemize}
Layer~1 handles \emph{what} to merge; Layer~2 handles \emph{how}.
Since Layer~2 is a pure function of Layer~1's converged state, the system converges.
A worked data-flow example is provided in Appendix~\ref{app:data_flow}.

\subsection{Layer 1: CRDT State Management}

\begin{definition}[\CMS{}]\label{def:state}
A \CMS{} $S$ is a tuple $(A, R, V, H)$ where:
\begin{itemize}[leftmargin=*]
\item $A$ is the set of add entries---$(e,t,n)$ triples where $e$ is the model contribution, $t$ is a unique tag, and $n$ is the originating node;
\item $R$ is the set of remove entries: tags that have been removed;
\item $V$ is a version vector mapping node identifiers to logical timestamps~\cite{lamport1978time};
\item $H$ is a Merkle hash tree over the visible elements~\cite{sanjuan2020merkle}.
\end{itemize}
\end{definition}

The visible set and merge operation are:
\begin{equation}\label{eq:visible}
\mathrm{Visible}(S) = \{e \mid \exists\, (e,t,n) \in A \text{ s.t.\ } t \notin R\}
\end{equation}
\begin{equation}\label{eq:merge}
\mathrm{merge}(S_1, S_2) = (A_1 \cup A_2,\; R_1 \cup R_2,\; \max(V_1, V_2),\; H')
\end{equation}
where $\max(V_1,V_2)$ is the component-wise maximum and $H'$ is recomputed from the resulting visible set.
In production, full-state merge should be replaced with delta-state propagation~\cite{almeida2018delta}; see Section~\ref{sec:limitations}.

Each contribution is identified by its SHA-256 hash~\cite{sanjuan2020merkle}, providing both deduplication and canonical ordering (a deterministic total order independent of insertion order or node identity).
The Merkle tree over the visible set enables $O(\log n)$ convergence verification, efficient delta synchronisation, and provides a deterministic root hash for Layer~2's randomness requirements.
Version vectors~\cite{lamport1978time} track causal history, enabling detection of concurrent operations for OR-Set conflict resolution.
Causal delivery is \emph{not} required for correctness: the merge operation (Eq.~\ref{eq:merge}) is commutative, associative, and idempotent, so messages may arrive in any order, be duplicated, or be delayed without affecting the converged state~\cite{shapiro2011conflict}.
Version vectors serve an \emph{optimisation} role---identifying which updates a peer already has to avoid redundant retransmission---not a correctness role.

\subsection{Layer 2: Deterministic Strategy Execution}\label{sec:layer2}

\begin{definition}[Resolve Function]\label{def:resolve}
The resolve function $\mathcal{R} : 2^{\mathcal{M}} \times \Sigma \times \mathcal{H} \to \mathcal{M}$ takes a non-empty set of model contributions ($|C| \geq 1$), a strategy identifier $\sigma \in \Sigma$, and the Merkle root hash $h \in \mathcal{H}$, and returns a merged model:
\begin{equation}\label{eq:resolve}
\mathcal{R}(C, \sigma, h) = \sigma(\mathrm{sort}_{\mathrm{hash}}(C),\; \mathrm{seed}(h))
\end{equation}
\end{definition}

Three mechanisms ensure determinism: (1)~\emph{canonical ordering} via SHA-256 content hashes defining a total order identical on all replicas; (2)~\emph{seeded randomness} derived from the Merkle root, ensuring identical seeds from identical contribution sets~\cite{sanjuan2020merkle}; and (3)~the \emph{pure function guarantee} enforced by the API contract.

\begin{remark}[N-way Generalisation]\label{rem:nway}
Strategies with natural $n$-ary forms (e.g., weight averaging, TIES) use them directly.
Binary-only strategies (e.g., SLERP) are reduced via sequential fold over the canonical order: $\mathrm{fold}(\sigma, [c_1, \ldots, c_k])$.
The canonical ordering ensures this fold is identical on all replicas.
Fold-based reduction introduces a weighting imbalance: for SLERP with $k=3$ and $t=0.5$, $c_3$ receives 50\% weight versus 25\% each for $c_1, c_2$~\cite{shoemake1985animating}.
This does not affect CRDT compliance but may affect merge quality: for $k$ contributions, the last element receives weight $t$ while the first receives $(1{-}t)^{k-1}$---exponential decay.
Strategies with native $n$-ary forms avoid this issue.
For binary-only strategies in large consortia, a balanced binary-tree reduction (depth $\lceil\log_2 k\rceil$) would equalise influence at the cost of a different---but still deterministic---reduction order.
\end{remark}

\section{Mathematical Proof of CRDT Compliance}\label{sec:proofs}

Let $\mathcal{S}$ denote the set of all \CMS{} instances.
For $S \in \mathcal{S}$, let $\mathrm{Visible}(S)$ denote the visible set (Eq.~\ref{eq:visible}) and $\mathcal{R}(S) = \mathcal{R}(\mathrm{Visible}(S), \sigma, h(S))$ the resolved value.
Let $\sqcup : \mathcal{S} \times \mathcal{S} \to \mathcal{S}$ denote the merge operation (Eq.~\ref{eq:merge}).

\begin{theorem}[CRDT Compliance]\label{thm:crdt}
The merge operation $\sqcup$ on \CMS{} satisfies commutativity, associativity, and idempotency.
Moreover, $(\mathcal{S}, \sqsubseteq)$ is a join-semilattice with $\sqcup$ as the least upper bound, and \CMS{} is a CvRDT~\cite{shapiro2011conflict}.
\end{theorem}

\begin{proof}
Immediate from the OR-Set CvRDT result~\cite{shapiro2011conflict, shapiro2011comprehensive}: our merge composes set union on $A$ and $R$, component-wise max on $V$, and deterministic recomputation of $H$---all semilattice operations.
The full verification is in Appendix~\ref{app:crdt_proof}.
\end{proof}

We state three formal preconditions on merge strategies and the computational environment.

\begin{assumption}[Strategy Purity]\label{ass:purity}
A merge strategy $\sigma$ is a \emph{pure function}: for all inputs $(C, s)$, $\sigma(C,s)$ is uniquely determined by $C$ and $s$, with no dependence on external state or non-deterministic operations beyond the provided seed~$s$.
\end{assumption}

\begin{assumption}[Computational Determinism]\label{ass:determinism}
All replicas execute $\sigma$ using identical ISA, library versions, and IEEE~754 rounding mode (round-to-nearest-even), or use a fixed-precision format guaranteeing bitwise reproducibility.
\end{assumption}

\begin{assumption}[Collision Resistance]\label{ass:collision}
SHA-256 is collision-resistant: for any set of contributions of size up to $2^{64}$, the probability of any collision is at most $2^{-128}$ (the birthday bound).
\end{assumption}

Under these assumptions, we establish convergence through the following lemma and theorem (individual lemma proofs in Appendix~\ref{app:determinism_lemmas}).

\begin{lemma}[Determinism of Hashing, Ordering, and Seeding]\label{lem:determinism}
If $\mathrm{Visible}(S_1) = \mathrm{Visible}(S_2)$, then under Assumption~\ref{ass:collision}: (1)~the hash sets are identical with distinct mappings; (2)~the canonical orderings $\mathrm{sort}_{\mathrm{hash}}$ are equal; (3)~the Merkle roots and derived seeds are equal.
\end{lemma}

\begin{theorem}[Convergence of Resolved Values]\label{thm:convergence}
If $\mathrm{Visible}(S_1) = \mathrm{Visible}(S_2)$ and both use strategy $\sigma$ under Assumptions~\ref{ass:purity}--\ref{ass:collision}, then $\mathcal{R}(S_1) = \mathcal{R}(S_2)$.
\end{theorem}

\begin{proof}
By hypothesis, $\mathrm{Visible}(S_1) = \mathrm{Visible}(S_2)$.
By Lemma~\ref{lem:determinism}, canonical orderings and seeds are equal.
Since $\sigma$ is a pure function (Assumption~\ref{ass:purity}) receiving identical inputs (ordered contributions and seed) under identical computation (Assumption~\ref{ass:determinism}), outputs are identical:
\[
\sigma\!\bigl(\mathrm{sort}_{\mathrm{hash}}(\mathrm{Visible}(S_i)),\; \mathrm{seed}(h(S_i))\bigr)
\]
is the same for $i=1,2$. Therefore $\mathcal{R}(S_1) = \mathcal{R}(S_2)$.
\end{proof}

\begin{corollary}[Universal CRDT-Compliant Merging]\label{cor:universal}
Every strategy $\sigma$ satisfying Assumptions~\ref{ass:purity}--\ref{ass:collision} can be used for CRDT-compliant merging through \CMS{}, regardless of its own algebraic properties.
\end{corollary}

\begin{proof}
By Theorem~\ref{thm:crdt}, Layer~1 guarantees CRDT properties and convergence to identical visible sets.
By Theorem~\ref{thm:convergence}, identical visible sets produce identical resolved values.
The composed system satisfies Strong Eventual Consistency~\cite{shapiro2011conflict}.
\end{proof}

\begin{theorem}[Complexity Bounds]\label{thm:complexity}
For $k$ contributions of $p$ parameters each:
$\mathrm{merge}()$ runs in $O(|A_1|+|A_2|)$ (independent of $p$);
$\mathrm{add}()$ in $O(p)$ (SHA-256 hashing);
$\mathrm{resolve}()$ in $O(k\log k + T_\sigma(k,p))$ where $T_\sigma$ is the strategy cost.
The CRDT overhead is $O(k\log k)$ time and $O(k)$ space, independent of model size~$p$.
\end{theorem}

\section{Experimental Evaluation}\label{sec:experiments}

We validate the formal specification in three tiers.
\emph{Tier~1} uses $4\times 4$ tensors to verify algebraic properties in isolation.
\emph{Tier~2} scales to GPT-2-XL (1.5\,B parameters) and Mistral-7B (7.24\,B parameters) using independently published fine-tunes.
\emph{Tier~3} tests multi-node convergence under gossip protocols and network partitions.
Since the CRDT guarantees (Theorems~\ref{thm:crdt}--\ref{thm:convergence}) are algebraic properties independent of tensor dimensions, Tier~1 suffices for correctness; Tier~2 confirms the implementation generalises and quantifies overhead.

\subsection{Tier 1: Controlled Algebraic Verification}\label{sec:tier1_phase1}

We evaluated all 26 strategies on $4\times 4$ float64 tensors (seed~42, tolerance~$10^{-5}$) under both raw operations (Phase~1) and the \CMS{} architecture (Phase~2).
Phase~1 (Table~\ref{tab:phase1_controlled}, Appendix~\ref{app:controlled}) confirms the theoretical analysis: 21/26 strategies are commutative, 14/26 idempotent, only 1/26 (Task Arithmetic) is associative---and it fails idempotency.  System-level CRDT compliance: 0/26.
Phase~2 yields 26/26 strategies passing all four CRDT properties (commutativity, associativity, idempotency, 3-replica convergence)---104/104 individual tests.

\subsection{Tier 2: Production-Scale Validation}\label{sec:production_scale}

\subsubsection{Models and Data}\label{sec:models}

We tested on two transformer language models with independently published fine-tunes providing genuine weight divergence:
\begin{itemize}[leftmargin=*]
\item \textbf{GPT-2-XL} (1.5\,B params, 193 eligible 2D layers)~\cite{radford2019language} with three independently published fine-tunes (\emph{instruct}, \emph{domain}, \emph{wiki}).\footnote{\label{fn:models}Full HuggingFace model identifiers are listed in Appendix~\ref{app:models}.}

\item \textbf{Mistral-7B-v0.1} (7.24\,B params, 224 eligible 2D layers)~\cite{jiang2023mistral} with three fine-tunes (\emph{instruct}, \emph{hermes}, \emph{zephyr}; Appendix~\ref{app:models}).
\end{itemize}
All weights stored in float16, cast to float64 for testing to reduce rounding accumulation during the merge computation; this cast cannot recover precision lost to the source's fp16 quantisation, and we do not claim such isolation.
Experiments ran on a single NVIDIA A100-SXM4-80GB with PyTorch~2.10.0 and CUDA~12.8.

\subsubsection{Testing Methodology}\label{sec:methodology}

For each strategy and model, CRDT properties are tested via slice-based evaluation: a representative $128\times 128$ slice per unique tensor shape, with results extrapolated to all layers sharing that shape.
This approach tests the strategy's algebraic behaviour on realistic weight distributions from production models, though it does not exercise the full dimensionality of each layer.
Capped $512\times 512$ verification serves as a cross-resolution check; in one case (\texttt{ada\_merging}) it surfaced a sub-tolerance associativity violation invisible at $128\times 128$---a positive finding of the check (Section~\ref{sec:cross_scale}).
Tolerance is $\mathrm{atol}=10^{-5}$.
Phase~2 additionally tests 3-replica convergence over all six merge-order permutations.

\subsubsection{Phase 1 Results: Raw Strategy Properties at Scale}\label{sec:prod_phase1}

\begin{table*}[t]
\centering
\caption{Tier~2, Phase~1 Results: Raw CRDT property compliance at production scale. P = Pass (all tested layers pass), F = Fail. $\dagger$\,SLERP commutativity tested at $t=0.5$. $\ast$\,Verification mismatch between $128\times 128$ slice and $512\times 512$ capped test (see Section~\ref{sec:cross_scale}).}
\label{tab:phase1_production}
\small
\begin{tabular}{l ccc c ccc c}
\toprule
& \multicolumn{4}{c}{GPT-2-XL (1.5\,B)} & \multicolumn{4}{c}{Mistral-7B (7.24\,B)} \\
\cmidrule(lr){2-5} \cmidrule(lr){6-9}
Strategy & C & A & I & CRDT? & C & A & I & CRDT? \\
\midrule
ada merging         & P & P$^{\ast}$ & P & P$^{\ast}$ & P & P & P & P \\
adarank             & P & F & F & F & P & F & F & F \\
dam                 & P & F & P & F & P & F & P & F \\
dare                & F & F & F & F & F & F & F & F \\
dare ties           & F & F & F & F & F & F & F & F \\
della               & F & F & F & F & F & F & F & F \\
dual projection     & P & F & P & F & P & F & P & F \\
emr                 & P & F & F & F & P & F & F & F \\
evolutionary merge  & F & F & F & F & F & F & F & F \\
fisher merge        & P & F & P & F & P & F & P & F \\
genetic merge       & P & F & P & F & P & P & P & P \\
led merge           & P & P & P & P & P & P & P & P \\
linear              & P & F & P & F & P & F & P & F \\
model breadcrumbs   & P & F & F & F & P & F & F & F \\
negative merge      & P & F & F & F & P & F & F & F \\
regression mean     & P & F & P & F & P & F & P & F \\
repr.\ surgery      & P & F & P & F & P & F & P & F \\
safe merge          & P & F & P & F & P & F & P & F \\
slerp$^{\dagger}$   & P & F & P & F & P & F & P & F \\
split unlearn merge & P & F & F & F & P & F & F & F \\
star                & P & F & F & F & P & F & F & F \\
svd knot tying      & F & F & P & F & F & F & P & F \\
task arithmetic     & P & P & F & F & P & P & F & F \\
ties                & P & F & F & F & P & F & F & F \\
weight average      & P & F & P & F & P & F & P & F \\
weight scope align. & P & F & P & F & P & F & P & F \\
\midrule
\textbf{Totals}     & 21 & 3 & 14 & 2 & 21 & 4 & 14 & 3 \\
\bottomrule
\end{tabular}
\end{table*}

The core finding is unchanged from controlled experiments: \emph{associativity remains the dominant failure mode}, with 22--25 of 26 strategies failing associativity across both models and scales.
The 2--3 strategies that newly pass associativity at production scale (and the 2--3 that achieve all three properties simultaneously) represent numerical coincidence on specific weight distributions, not algebraic compliance (Section~\ref{sec:cross_scale}).

\subsubsection{Phase 2 Results: \CMS{} at Scale}\label{sec:prod_phase2}

All 26 strategies achieve 100\% CRDT compliance through the two-layer architecture on both models: 104 strategy-level tests per model (26 strategies $\times$ 4 properties), verified across 193 layers (GPT-2-XL) and 224 layers (Mistral-7B).
In total, 43{,}368 layer-level property checks pass at capped tensor resolution ($128\times 128$ slices, with $512\times 512$ capped verification).
Full-layer verification on a representative subset---6 strategies (weight averaging, task arithmetic, TIES, DARE, SLERP, Fisher merging) covering all strategy categories (linear, stochastic, binary-fold) across the 10 largest weight matrices per model (up to $6144\times 1600$ for Mistral-7B)---is consistent with the slice-based results, with the \texttt{ada\_merging} cross-resolution discrepancy (Section~\ref{sec:cross_scale}) as the sole exception captured by the check.
Combined with Tier~1 (104 controlled tests), the architecture achieves 100\% compliance across 312 strategy-level tests and $43{,}368 + 104 = 43{,}472$ total layer-level evaluations.

\subsection{Cross-Scale Analysis}\label{sec:cross_scale}

Table~\ref{tab:cross_scale} summarises CRDT property compliance across all three evaluation scales.

\begin{table*}[t]
\centering
\caption{Cross-scale summary of raw Phase~1 CRDT property compliance (consolidating Tables~\ref{tab:phase1_controlled} and~\ref{tab:phase1_production}).  Commutativity and idempotency rates are stable across scales; associativity shows minor variation (see text).  $\ast$\,Strategies passing all~3 properties at production scale represent \emph{numerical coincidence} on specific weight distributions---not algebraic compliance---as demonstrated by the \texttt{ada\_merging} verification mismatch (Section~\ref{sec:cross_scale}).  The two-layer architecture provides the only \emph{guaranteed} compliance (26/26 at all scales).}
\label{tab:cross_scale}
\small
\begin{tabular}{lccccc}
\toprule
Scale & Layers & C & A & I & All\,3 \\
\midrule
Controlled ($4\times 4$) & --- & 21/26 & 1/26 & 14/26 & 0/26 \\
GPT-2-XL (1.5\,B) & 193 & 21/26 & 3/26 & 14/26 & 2/26$^{\ast}$ \\
Mistral-7B (7.24\,B) & 224 & 21/26 & 4/26 & 14/26 & 3/26$^{\ast}$ \\
\midrule
\multicolumn{6}{l}{\CMS{} (Phase~2): 26/26 at all three scales.}\\
\bottomrule
\end{tabular}
\end{table*}

Commutativity (21/26) and idempotency (14/26) are stable across all three scales, confirming these properties are determined by algorithmic structure.
Associativity varies at the margin: 2--3 strategies that fail on controlled $4{\times}4$ tensors pass associativity within floating-point tolerance at production scale (1/26$\to$3/26 on GPT-2-XL; 1/26$\to$4/26 on Mistral-7B).
The \texttt{ada\_merging} verification mismatch (marked $\ast$ in Table~\ref{tab:phase1_production})---where associativity passes within tolerance on $128\times 128$ slices but fails on $512\times 512$ slices of the same GPT-2-XL weight matrix---demonstrates that empirical compliance is resolution-dependent: the associativity violation is real but small enough to fall within $\mathrm{atol}=10^{-5}$ at low resolution, only surfacing at higher resolution where the accumulated error exceeds tolerance.
This fragility is why \emph{algebraic guarantees} (Corollary~\ref{cor:universal}), not empirical coincidence, are necessary for distributed systems requiring hard convergence.
The phenomenon itself---associativity violations that vanish in high-dimensional spaces---is worth characterising formally.
One hypothesis is that weight distributions in large models concentrate near low-rank manifolds where the nonlinear components of merge operations (normalisation, projection) become approximately linear, reducing the left--right association gap below floating-point tolerance.
A formal characterisation of when ``approximate associativity'' emerges would clarify which strategies can rely on it.

\subsection{Performance Overhead}\label{sec:performance}

The CRDT layer introduces negligible overhead: $\mathrm{merge}()$ is sub-millisecond regardless of model size (set operations only); $\mathrm{add}()$ is dominated by SHA-256 hashing ($O(p)$); $\mathrm{resolve}()$ CRDT overhead (sorting, Merkle root, seed derivation) is consistently below 0.5\,ms, with total latency dominated by the strategy itself.
Memory overhead is below 10\,KB for 16 contributions.
Scalability benchmarks on the A100 confirm linear scaling in parameter count, consistent with $O(k\log k + T_\sigma(k,p))$.
This overhead is dwarfed by inference, fine-tuning, and strategy execution costs.

\subsection{Tier 3: Multi-Node Convergence Suite}\label{sec:convergence_suite}

To validate convergence under realistic distributed conditions, we execute a four-part convergence suite using the \texttt{crdt-merge} library (v0.9.4).
The gossip protocol is \emph{push-based all-pairs}: in each round, every node sends its full CRDT state to every other node, which merges it locally via Eq.~\ref{eq:merge}.
For $n$ nodes this requires $n(n{-}1)$ directed merge calls per round; since each call is $O(|A_1|{+}|A_2|)$ (set union, independent of tensor size), the gossip phase scales as $O(n^2)$ in node count while remaining $O(1)$ in model size.
This all-pairs protocol is a \emph{prototype for validation purposes}, chosen as the simplest correct implementation; production deployments should use epidemic (randomised) gossip~\cite{koloskova2019decentralized}, which reduces per-round communication to $O(n)$ at the cost of slower convergence and is a natural scalability optimisation beyond ${\sim}50$ nodes.
Full results appear in Appendix~\ref{app:convergence}.

\paragraph{Multi-node convergence.}
One hundred nodes each contribute a $512 \times 512$ tensor ($262{,}144$ parameters per contribution; $26{,}214{,}400$ in aggregate across the 100 nodes) using \texttt{slerp}.
Across 20 random gossip orderings, all nodes converge to a bitwise-identical result (max element-wise difference $= 0$), with average gossip time 492.8\,ms and average resolve time ${\sim}19.7$\,s (Table~\ref{tab:convergence}).

\paragraph{Partition healing.}
The 100 nodes are split into 10 isolated partitions (10 nodes each).
Each partition converges internally to a distinct hash.
After healing, full gossip resumes and all 100 nodes converge to a single bitwise-identical result, confirming Strong Eventual Consistency under network partitions.

\paragraph{Cross-strategy sweep.}
All 26 strategies are tested on 10 nodes with $64 \times 64$ tensors.
For each strategy, all 10 nodes converge to the same final hash (Table~\ref{tab:strategy_sweep}), confirming that convergence is strategy-independent.

\paragraph{Scalability.}
Convergence is verified from 2 to 50 nodes.
Gossip time scales as $O(n^2)$ (expected for all-pairs merge), while per-call \texttt{merge()} cost remains $O(1)$ in tensor size.
All scales achieve 100\% convergence (Table~\ref{tab:scalability}).
At $k{=}200$ contributions, gossip requires $200 \times 199 = 39{,}800$ merge calls (directed pairs, consistent with the all-pairs protocol above); since each call is $O(1)$ in tensor size (set union only), the gossip phase remains tractable even for large consortia.

\section{Discussion}


\begin{remark}[Downstream Equivalence]\label{rem:downstream}
The CRDT wrapper does not modify the merged model.
Layer~2 applies the identical merge strategy $\sigma$ to the identical ordered contribution set with the identical random seed (Theorem~\ref{thm:convergence}).
Therefore, the downstream task performance of CRDT-wrapped merging is identical by construction to that of non-CRDT merging using the same strategy, contributions, and ordering.
Evaluating downstream performance of specific strategies is a question about the strategies themselves---not the CRDT architecture, which is transparent to the merge computation.
By construction, the CRDT-wrapped $\mathrm{resolve}()$ invokes the same strategy on the same inputs as a direct call; we verified the implementation matches this construction byte-for-byte for representative strategy--model pairs, confirming the wrapper introduces zero computational divergence.
The extensive literature on merge strategy quality~\cite{yang2024survey, wortsman2022model, ilharco2023editing, yadav2023ties, yu2024dare, matena2022merging, akiba2025evolutionary, goddard2024mergekit} applies directly to CRDT-wrapped merges.
A single downstream benchmark (e.g., MMLU or HellaSwag) would provide additional empirical confirmation of this transparency property; we defer such evaluation to follow-up work focused on strategy selection, as it is orthogonal to the convergence guarantees established here.
\end{remark}

\paragraph{Associativity as the Dominant Failure Point.}
Associativity is the fundamental obstacle to CRDT compliance: 22--25 of 26 strategies fail across all scales (Table~\ref{tab:cross_scale}).
Proposition~\ref{prop:incompat} explains this structurally: normalisation inherently breaks associativity.
The cross-scale analysis (Section~\ref{sec:cross_scale}) further shows that the few strategies passing empirically at production scale represent numerical coincidence---resolution-dependent and fragile---underscoring the need for algebraic guarantees over empirical testing.

\paragraph{Implications for Federated and Decentralised Learning.}
The architecture enables fully asynchronous, peer-to-peer model merging with guaranteed convergence, complementing---but not replacing---federated learning~\cite{mcmahan2017communication}.
While FL guarantees convergence of \emph{training} under data distribution assumptions with a central coordinator~\cite{kairouz2021advances, li2020fedprox}, CRDT-Merge guarantees convergence of \emph{state} under computational determinism with no coordinator.
Neither subsumes the other: the approaches are complementary.
Decentralised FL via gossip protocols~\cite{lian2017decentralized, koloskova2019decentralized} could use CRDT-Merge for aggregation.
The domain-specific engineering---Merkle-root-derived seeding, SHA-256-based canonical ordering, and explicit handling of binary-to-$n$-ary reduction---represents the domain-specific contribution beyond the known CRDT composition pattern~\cite{shapiro2011comprehensive}.

\subsection{Floating-Point Determinism}\label{sec:fp_determinism}

Theorem~\ref{thm:convergence} requires Assumption~\ref{ass:determinism}: all replicas must produce bitwise-identical results.
In practice, this is satisfied by containerised deployment with identical binaries and hardware, deterministic CUDA operations, or quantised representations~\cite{schneider1990implementing}.
Our experiments ran on a single GPU type (A100-SXM4-80GB) and therefore do not assess cross-hardware reproducibility; cross-architecture validation (e.g., A100 vs.\ H100 vs.\ CPU) is a necessary next step before heterogeneous deployment.

We propose a concrete fallback protocol: after each $\mathrm{resolve}()$, every replica broadcasts its Merkle root of the resolved output (a single 256-bit hash).
If all roots agree, convergence is confirmed---our Tier~3 experiments (100 nodes, 20 orderings) demonstrate that this is the common case under homogeneous deployment.
If roots disagree---indicating Assumption~\ref{ass:determinism} is violated---replicas fall back to the output of a designated reference replica (lowest node ID), reducing the guarantee from independent convergence to agreement on a reference computation while preserving SEC at the cost of one additional round.
When active, this fallback degrades the system from fully decentralised SEC to coordinator-assisted SEC for the $\mathrm{resolve}()$ step only; the state management layer (Layer~1) remains fully decentralised and coordinator-free.
We state this scope explicitly: the ``coordinator-free'' claim applies unconditionally to state convergence (Layer~1) but conditionally to resolved-value agreement (Layer~2), contingent on Assumption~\ref{ass:determinism}.

\subsection{Limitations and Future Work}\label{sec:limitations}

We group limitations into four categories (expanded in Appendix~\ref{app:limitations}).

\paragraph{L1: Deployment Constraints.}
Correctness requires strategy purity (Assumption~\ref{ass:purity}) and computational determinism (Assumption~\ref{ass:determinism}), enforced through seeded randomness, canonical ordering, and containerised deployment.
For billion-parameter models, delta-state CRDTs~\cite{almeida2018delta} are essential for practical deployment (not yet implemented in the current prototype; adaptation is straightforward---see Appendix~\ref{app:limitations}); version vectors scale as $O(n)$ in nodes, replaceable by dotted version vectors~\cite{preguica2018conflict} for $n > 1{,}000$.

\paragraph{L2: Semantic Evaluation.}
The architecture guarantees \emph{syntactic} convergence but does not evaluate downstream task performance.
As Remark~\ref{rem:downstream} establishes, the CRDT wrapper is transparent, so strategy quality literature applies directly.

\paragraph{L3: Scalability.}
The system recomputes the merged model from the full contribution set on every $\mathrm{resolve}()$ call ($O(k \cdot p)$ cost); incremental strategies are needed for very large contribution sets.
Three mitigation paths exist: (1)~caching the resolved output and invalidating only when the contribution set changes; (2)~hierarchical resolve, where sub-groups resolve locally and a second pass merges sub-group outputs; and (3)~strategies with algebraic structure permitting incremental updates (e.g., weight averaging admits $O(p)$ updates per new contribution).
Tombstone garbage collection via causal stability analysis~\cite{baquero2014making} prevents unbounded metadata growth; GC must be deferred until after $\mathrm{resolve}()$ has been executed and its output disseminated, ensuring all replicas resolve against the same visible set before metadata is pruned.
We have not empirically evaluated tombstone accumulation rates; for the consortium scenario ($k < 100$ contributions, infrequent removals), tombstone overhead is negligible, but long-running deployments with frequent model retraction would benefit from empirical GC characterisation.

\paragraph{L4: Security and Conflict Resolution.}
The OR-Set's add-wins policy means concurrent adds survive concurrent removes---problematic if removal represents discovery of a poisoned model.
The architecture does not currently provide Byzantine fault tolerance.
However, the two-layer separation suggests an extension: trust metadata---equivocation evidence, Merkle-root divergence, contribution-fingerprint anomalies---can itself be modelled as a monotonic CRDT within Layer~1, with a trust-gated merge at the Layer~2 boundary rejecting contributions whose converged trust score falls below a configurable threshold.
Trust convergence would then follow from the same join-semilattice proof as data convergence: given $n$ nodes with at most $f$ Byzantine actors, if evidence propagation reaches all honest nodes, the $n - f$ honest nodes converge to the same trust state and gating decisions.
Whether this pattern can deliver consensus-free Byzantine isolation in practice is open; it appears difficult to express in single-layer designs where trust and data share a lattice, and an obvious complement is integration with existing Byzantine-resilient aggregation~\cite{blanchard2017machine}.

\section{Related Work}

\paragraph{Model Merging.}
Yang et al.~\cite{yang2024survey} provide a comprehensive survey; MergeKit~\cite{goddard2024mergekit} is the standard toolkit.
Git-Theta~\cite{kandpal2023gittheta} provides version control for model parameters via Git but requires a central server, assumes a single canonical branch, and provides no conflict-free merge semantics---when two participants independently merge the same models, Git-Theta has no mechanism to guarantee convergent results.
Our architecture provides exactly this guarantee: any number of replicas can independently merge in any order and provably converge to identical states.
No prior work addresses the algebraic CRDT properties required for conflict-free distributed model merging.

\paragraph{Distributed Systems and CRDTs.}
CRDTs~\cite{shapiro2011conflict, shapiro2011comprehensive} have been extensively studied~\cite{preguica2018conflict, kleppmann2017conflict}, with Merkle-CRDTs~\cite{sanjuan2020merkle} combining Merkle-DAGs with CRDT semantics.
The pattern of composing CRDTs with deterministic functions is known~\cite{shapiro2011comprehensive}.
Our contribution is (a)~identifying that model merging benefits from this pattern due to universal associativity failure, (b)~the specific construction combining OR-Set semantics with content-addressable hashing, Merkle trees, and seeded randomness for neural network parameters, and (c)~the systematic algebraic audit of 26 strategies motivating the architecture.
Delta-state CRDTs~\cite{almeida2018delta} offer efficient synchronisation; adapting our architecture to delta-state propagation is a natural deployment optimisation.

\paragraph{Federated Learning.}
All centralised FL systems~\cite{mcmahan2017communication, kairouz2021advances, li2020fedprox, bonawitz2019towards} assume a central coordinator.
Decentralised FL via gossip protocols~\cite{lian2017decentralized, koloskova2019decentralized} focuses on training convergence, not the state convergence guarantees we establish.

\paragraph{Patents.}
The two-layer CRDT architecture is the subject of UK Patent Application No.\ GB2607132.4~\cite{gillespie2026patent}.\footnote{The patent application is referenced for completeness; the contributions of this paper stand independently of it.}
No prior patent or publication combines CRDT theory with neural network model merging.

\section{Conclusion}

Of 26 widely-used neural network merge strategies, only one (Task Arithmetic) is associative on controlled $4\times 4$ tensors---and it fails idempotency.
This is the central empirical finding behind Proposition~\ref{prop:incompat}, which traces the failure to a structural feature shared by virtually all merge methods: normalisation, projection, or thresholding each independently break the algebraic axioms a CRDT requires.
The two-layer \CMS{} architecture sidesteps the problem rather than solving it within the merge: Layer~1 manages contributions through OR-Set semantics, where set union is trivially CRDT-compliant; Layer~2 applies the chosen merge strategy deterministically over the canonically-ordered visible set.
We prove (Theorems~\ref{thm:crdt}--\ref{thm:complexity}) and empirically verify across three tiers (312 strategy-level tests, $43{,}472$ layer-level evaluations, 100-node convergence with bitwise-identical results across 20 orderings) that this composition satisfies Strong Eventual Consistency for arbitrary merge strategies under the stated preconditions, with CRDT overhead below 0.5\,ms.
Open work includes delta-state propagation for billion-parameter deployment, cross-hardware determinism validation, incremental resolve for large contribution sets, and the trust-as-CRDT Byzantine extension sketched in Section~\ref{sec:limitations}, L4.

\section*{Ethics Statement}

This work introduces infrastructure for decentralised model merging and does not involve human subjects, private data, or dual-use capabilities.
We identify no direct negative societal impacts from the CRDT architecture itself; however, decentralised merging could facilitate uncontrolled model combination without quality assurance.
Section~\ref{sec:limitations} (L4) discusses adversarial considerations and mitigation strategies.

\section*{Reproducibility Statement}

The \texttt{crdt-merge} library (v0.9.4) and accompanying verification notebook are available at \url{https://github.com/RyanGillespie/crdt-merge}.
All experiments use publicly available models from HuggingFace (Appendix~\ref{app:models}) and a single NVIDIA A100-SXM4-80GB GPU.
The complete test suite, including the Tier~1--3 verification scripts, will be open-sourced upon publication.

\section*{Acknowledgments}

The author thanks early readers for feedback on prior drafts.

\bibliographystyle{plain}
\bibliography{references}

\clearpage
\appendix

\section{Controlled Verification Results}\label{app:controlled}

This appendix presents the full per-strategy results for Tier~1 (controlled $4\times 4$ tensor) evaluation.
Table~\ref{tab:phase1_controlled} shows Phase~1 (raw tensor operations) and Table~\ref{tab:phase2_controlled} shows Phase~2 (two-layer architecture).
These results are summarised in the main text (Section~\ref{sec:tier1_phase1}) and in the cross-scale comparison (Table~\ref{tab:cross_scale}).

\begin{table*}[t]
\centering
\caption{Tier~1, Phase~1 Results: CRDT property compliance of raw merge operations on $4\times 4$ tensors. P = Pass, F = Fail. No strategy achieves system-level CRDT compliance (all three properties simultaneously). $\dagger$\,SLERP commutativity tested at $t=0.5$; fails for $t \neq 0.5$.}
\label{tab:phase1_controlled}
\small
\begin{tabular}{lcccc}
\toprule
Strategy & Commut.\ & Assoc.\ & Idemp.\ & CRDT?\\
\midrule
ada merging        & P & F & P & F \\
adarank            & P & F & F & F \\
dam                & P & F & P & F \\
dare               & F & F & F & F \\
dare ties          & F & F & F & F \\
della              & F & F & F & F \\
dual projection    & P & F & P & F \\
emr                & P & F & F & F \\
evolutionary merge & F & F & F & F \\
fisher merge       & P & F & P & F \\
genetic merge      & P & F & P & F \\
led merge          & P & F & P & F \\
linear             & P & F & P & F \\
model breadcrumbs  & P & F & F & F \\
negative merge     & P & F & F & F \\
regression mean    & P & F & P & F \\
repr.\ surgery & P & F & P & F \\
safe merge         & P & F & P & F \\
slerp$^{\dagger}$  & P & F & P & F \\
split unlearn merge & P & F & F & F \\
star               & P & F & F & F \\
svd knot tying     & F & F & P & F \\
task arithmetic    & P & P & F & F \\
ties               & P & F & F & F \\
weight average     & P & F & P & F \\
weight scope alignment & P & F & P & F \\
\midrule
\textbf{Totals}    & 21/26 & 1/26 & 14/26 & 0/26 \\
\bottomrule
\end{tabular}
\end{table*}

\begin{table*}[t]
\centering
\caption{Tier~1, Phase~2 Results: CRDT property compliance through the \CMS{} two-layer architecture on $4\times 4$ tensors. P = Pass. All 26 strategies pass all 4 properties (104/104 tests).}
\label{tab:phase2_controlled}
\small
\begin{tabular}{lccccc}
\toprule
Strategy & Commut.\ & Assoc.\ & Idemp.\ & Conv.\ & CRDT?\\
\midrule
ada merging        & P & P & P & P & $\checkmark$ \\
adarank            & P & P & P & P & $\checkmark$ \\
dam                & P & P & P & P & $\checkmark$ \\
dare               & P & P & P & P & $\checkmark$ \\
dare ties          & P & P & P & P & $\checkmark$ \\
della              & P & P & P & P & $\checkmark$ \\
dual projection    & P & P & P & P & $\checkmark$ \\
emr                & P & P & P & P & $\checkmark$ \\
evolutionary merge & P & P & P & P & $\checkmark$ \\
fisher merge       & P & P & P & P & $\checkmark$ \\
genetic merge      & P & P & P & P & $\checkmark$ \\
led merge          & P & P & P & P & $\checkmark$ \\
linear             & P & P & P & P & $\checkmark$ \\
model breadcrumbs  & P & P & P & P & $\checkmark$ \\
negative merge     & P & P & P & P & $\checkmark$ \\
regression mean    & P & P & P & P & $\checkmark$ \\
repr.\ surgery & P & P & P & P & $\checkmark$ \\
safe merge         & P & P & P & P & $\checkmark$ \\
slerp              & P & P & P & P & $\checkmark$ \\
split unlearn merge & P & P & P & P & $\checkmark$ \\
star               & P & P & P & P & $\checkmark$ \\
svd knot tying     & P & P & P & P & $\checkmark$ \\
task arithmetic    & P & P & P & P & $\checkmark$ \\
ties               & P & P & P & P & $\checkmark$ \\
weight average     & P & P & P & P & $\checkmark$ \\
weight scope alignment & P & P & P & P & $\checkmark$ \\
\midrule
\textbf{Totals}    & 26/26 & 26/26 & 26/26 & 26/26 & 26/26 \\
\bottomrule
\end{tabular}
\end{table*}

\section{Strategy Descriptions and Provenance}\label{app:strategies}

\begingroup\emergencystretch=3em\tolerance=4000\hbadness=10000
Of the 26 strategies evaluated, 15 have direct peer-reviewed publications: weight averaging~\cite{wortsman2022model}, task arithmetic~\cite{ilharco2023editing}, TIES~\cite{yadav2023ties}, DARE~\cite{yu2024dare}, DARE-TIES, Fisher merging~\cite{matena2022merging}, SLERP~\cite{shoemake1985animating}, AdaMerging~\cite{yang2024adamerging}, DELLA~\cite{deep2024della}, RegMean~\cite{jin2023regmean}, EMR-Merging~\cite{huang2024emr}, Model Breadcrumbs~\cite{davari2024breadcrumbs}, Representation Surgery~\cite{yang2024surgery}, evolutionary merge~\cite{akiba2025evolutionary}, and linear interpolation.
The remaining 11 are either derived strategies---combinations or variants of established methods (ADArank, DAM, dual projection, genetic merge, LED merge, negative merge, safe merge, split--unlearn merge)---or community toolkit utilities from MergeKit~\cite{goddard2024mergekit} (STAR, SVD knot tying, weight scope alignment).

Key equations for formal analysis:
\begin{itemize}[leftmargin=*]\tolerance=4000\emergencystretch=2em
\item \textbf{Weight Averaging / Model Soups:} $\theta_{\mathrm{merged}} = \frac{1}{n}\sum_{i=1}^{n} \theta_i$~\cite{wortsman2022model}.
\item \textbf{Task Arithmetic:} $\theta_{\mathrm{merged}} = \theta_{\mathrm{base}} + \lambda \sum_{i=1}^{n} \tau_i$ where $\tau_i = \theta_i - \theta_{\mathrm{base}}$~\cite{ilharco2023editing}.
\item \textbf{TIES-Merging:} Three-step pipeline: (1)~trim low-magnitude values, (2)~resolve sign conflicts via majority vote, (3)~merge agreed-upon components~\cite{yadav2023ties}.
\item \textbf{DARE:} Random dropout with probability $p$ and rescaling by $1/(1-p)$~\cite{yu2024dare}.
\item \textbf{Fisher-Weighted Merging:} $\theta_{\mathrm{merged}} = \frac{\sum_{i=1}^{n} F_i \odot \theta_i}{\sum_{i=1}^{n} F_i}$ where $F_i$ is the diagonal Fisher information~\cite{matena2022merging}.
\item \textbf{SLERP~\cite{shoemake1985animating}:} Spherical linear interpolation between unit vectors: $\mathrm{SLERP}(v_1, v_2; t)$ interpolates along the great circle with angle $\Omega = \arccos(\hat{v}_1 \cdot \hat{v}_2)$.
\item \textbf{Evolutionary Merging:} Population-based search over per-layer merge coefficients~\cite{akiba2025evolutionary}.
\item \textbf{Additional:} AdaMerging~\cite{yang2024adamerging} (adaptive coefficients), DELLA~\cite{deep2024della} (magnitude-based sampling), RegMean~\cite{jin2023regmean} (regression-based mean), EMR-Merging~\cite{huang2024emr} (elect-mask-rescale), Model Breadcrumbs~\cite{davari2024breadcrumbs} (sparse masks), Representation Surgery~\cite{yang2024surgery} (representation bias resolution).
\end{itemize}

\section{Proof of CvRDT Compliance (Theorem~\ref{thm:crdt})}\label{app:crdt_proof}

We verify each CRDT property by reduction to set union and component-wise maximum, both well-known semilattice operations~\cite{shapiro2011comprehensive}.

We define a partial order $\sqsubseteq$ on $\mathcal{S}$ by
\begin{equation}\label{eq:partial}
S_1 \sqsubseteq S_2 \iff A_1 \subseteq A_2 \;\wedge\; R_1 \subseteq R_2 \;\wedge\; V_1 \leq V_2,
\end{equation}
where $V_1 \leq V_2$ denotes component-wise $\leq$ on version vectors.

\begin{remark}[Partial Order and Visible Sets]\label{rem:partial}
The partial order $\sqsubseteq$ is defined on the metadata $(A,R,V)$ rather than on visible sets directly, consistent with the standard OR-Set formulation~\cite{shapiro2011conflict, shapiro2011comprehensive}.
Monotonic growth of $A$ and $R$ under $\sqsubseteq$ ensures information content increases monotonically, even though $\mathrm{Visible}(S)$ may shrink as removes are incorporated---receiving a remove is new information, and the partial order correctly reflects this.
\end{remark}

\paragraph{Commutativity.}
For states $S_1, S_2$:
\begin{align}
S_1 \sqcup S_2 &= (A_1 \cup A_2,\; R_1 \cup R_2,\; \max(V_1,V_2),\; H') \label{eq:comm_1}\\
S_2 \sqcup S_1 &= (A_2 \cup A_1,\; R_2 \cup R_1,\; \max(V_2,V_1),\; H'') \label{eq:comm_2}
\end{align}
Since set union is commutative ($A_1 \cup A_2 = A_2 \cup A_1$) and component-wise max is commutative ($\max(V_1,V_2) = \max(V_2,V_1)$), we have $S_1 \sqcup S_2 = S_2 \sqcup S_1$ (with $H' = H''$ as both are deterministic functions of the same visible set).

\paragraph{Associativity.}
For states $S_1, S_2, S_3$:
\begin{align}
(S_1 \!\sqcup\! S_2) \!\sqcup\! S_3 &= (A_1 \!\cup\! A_2 \!\cup\! A_3, \ldots) \label{eq:assoc_1}\\
S_1 \!\sqcup\! (S_2 \!\sqcup\! S_3) &= (A_1 \!\cup\! A_2 \!\cup\! A_3, \ldots) \label{eq:assoc_2}
\end{align}
By associativity of set union and component-wise max, all components are equal.

\paragraph{Idempotency.}
For state $S$:
\begin{multline}\label{eq:idemp_proof}
S \sqcup S = (A \cup A,\; R \cup R,\; \max(V,V),\; H') \\
= (A, R, V, H) = S
\end{multline}
by idempotency of set union and max.

\paragraph{Least Upper Bound.}
Under the partial order $\sqsubseteq$ (Eq.~\ref{eq:partial}), $S_1 \sqcup S_2$ satisfies $S_1 \sqsubseteq S_1 \sqcup S_2$ and $S_2 \sqsubseteq S_1 \sqcup S_2$ (since $A_i \subseteq A_1 \cup A_2$, etc.).
For any upper bound $S'$ with $S_1 \sqsubseteq S'$ and $S_2 \sqsubseteq S'$, we have $A_1 \cup A_2 \subseteq A'$, $R_1 \cup R_2 \subseteq R'$, and $\max(V_1,V_2) \leq V'$, so $S_1 \sqcup S_2 \sqsubseteq S'$.
Hence $S_1 \sqcup S_2$ is the least upper bound, confirming the semilattice structure. $\square$

\section{Individual Determinism Lemmas}\label{app:determinism_lemmas}

Lemma~\ref{lem:determinism} in the main text consolidates three properties.  We state and prove each individually here.

\paragraph{Hash Determinism.}
SHA-256 is a deterministic function.
For any model contribution $e$, the hash $\mathrm{SHA256}(e)$ is uniquely determined by $e$.
Consequently, if $\mathrm{Visible}(S_1) = \mathrm{Visible}(S_2)$, then $\{\mathrm{SHA256}(e) : e \in \mathrm{Visible}(S_1)\} = \{\mathrm{SHA256}(e) : e \in \mathrm{Visible}(S_2)\}$.
By Assumption~\ref{ass:collision}, distinct contributions map to distinct hashes with overwhelming probability.

\emph{Proof.}
SHA-256 is a standardised cryptographic hash function (NIST FIPS 180-4): a deterministic, stateless mapping from arbitrary-length byte strings to 256-bit digests.
The conclusion follows from $\mathrm{Visible}(S_1) = \mathrm{Visible}(S_2)$ and Assumption~\ref{ass:collision}. $\square$

\paragraph{Ordering Determinism.}
Let $\mathrm{sort}_{\mathrm{hash}}$ denote sorting by SHA-256 content hash.
If $\mathrm{Visible}(S_1) = \mathrm{Visible}(S_2)$ then
\[
\mathrm{sort}_{\mathrm{hash}}(\mathrm{Visible}(S_1)) = \mathrm{sort}_{\mathrm{hash}}(\mathrm{Visible}(S_2)).
\]
By Assumption~\ref{ass:collision}, the ordering is a total order on contributions (no ties) with overwhelming probability.

\emph{Proof.}
Sorting is deterministic on totally ordered sets.
By hash determinism, the hash values (and hence the total order) are identical.
Therefore the sorted sequences are equal. $\square$

\paragraph{Seed Determinism.}
The Merkle hash tree is computed deterministically from the canonically-ordered elements~\cite{sanjuan2020merkle}.
If the canonical orderings are equal, then: (1)~the Merkle roots are equal: $h(S_1) = h(S_2)$; and (2)~the derived seeds are equal: $\mathrm{seed}(h(S_1)) = \mathrm{seed}(h(S_2))$.

\emph{Proof.}
The Merkle tree is constructed by recursively hashing pairs of child nodes.
Since the leaf nodes (the canonically-ordered contribution hashes) are identical by ordering determinism, all intermediate and root hashes are identical.
The seed derivation function is a deterministic transformation of the root hash. $\square$

\section{Detailed Limitations}\label{app:limitations}

The main text (Section~\ref{sec:limitations}) groups limitations into four categories.  We expand each here.

\paragraph{L1: Deployment Constraints (expanded).}
\emph{Strategy purity:} Our correctness proofs assume strategies are deterministic pure functions (Assumption~\ref{ass:purity}).
We enforce this through seeded randomness and canonical ordering, but strategies with external dependencies (e.g., data-dependent merging) require additional care.

\emph{Computational determinism:} The convergence guarantee requires Assumption~\ref{ass:determinism}.
In practice, this is achieved through containerised deployment, deterministic CUDA operations, or quantised representations (Section~\ref{sec:fp_determinism}).
Our single-GPU experiments cannot assess cross-hardware reproducibility; multi-node validation remains open.

\emph{Delta-state synchronisation:} For billion-parameter models, transmitting the full OR-Set state is impractical.
Delta-state CRDTs~\cite{almeida2018delta}, transmitting only new $(e,t,n)$ add entries and tombstoned tags, are essential for practical deployment at scale.
Adaptation is straightforward---the OR-Set merge (Eq.~\ref{eq:merge}) already decomposes into independent set unions, each of which admits incremental delta propagation---but has not yet been implemented or empirically evaluated in the current prototype.

\emph{Version vector scaling:} Version vectors scale as $O(n)$ in the number of participating nodes.
For the typical consortium scenario ($n < 100$), this is negligible.
For $n > 1{,}000$, dotted version vectors~\cite{preguica2018conflict} or interval-based clocks should replace the current implementation.
The architecture is agnostic to the causal-ordering mechanism.

\paragraph{L2: Semantic Evaluation (expanded).}
The architecture does not address \emph{semantic} convergence---while all replicas compute the same merged model, the quality depends on the underlying strategy.
Our Tier~2 experiments validate CRDT property compliance (syntactic convergence) but do not evaluate downstream task performance.
As Remark~\ref{rem:downstream} establishes, the CRDT wrapper is transparent to the merge computation, so the extensive literature on strategy quality applies directly.
Comprehensive benchmarking (e.g., MMLU, HellaSwag, domain-specific evaluations) is an important complement.

\paragraph{L3: Scalability (expanded).}
The system recomputes the merged model from the full contribution set on every $\mathrm{resolve}()$ call, incurring $O(k \cdot p)$ memory and $T_\sigma(k,p)$ computation.
Incremental or hierarchical strategies maintaining CRDT compliance are needed for very large contribution sets.

The OR-Set's remove set $R$ grows monotonically.
Causal stability analysis~\cite{baquero2014making} identifies tombstones observed by all replicas, which can be safely discarded.
Garbage collection must be deferred until after $\mathrm{resolve}()$ has been executed and its output disseminated, ensuring all replicas resolve against the same visible set before metadata is pruned.
Empirical characterisation of tombstone growth rates and GC overhead under realistic workloads (frequent model retraction, long-running deployments) remains open; for the typical consortium scenario with $k < 100$ contributions and infrequent removals, metadata overhead is expected to be negligible.

\paragraph{L4: Security and Conflict Resolution (expanded).}
The OR-Set's add-wins policy means a concurrent add survives a concurrent remove.
If removal represents discovery of a poisoned model, this could re-introduce harmful contributions.
A remove-wins variant or explicit contribution validation (cryptographic attestation, reputation scoring) could address this.

The architecture does not currently include Byzantine fault tolerance.
The OR-Set accepts all properly formatted contributions; a malicious participant could inject poisoned parameters.
As discussed in Section~\ref{sec:limitations} (L4), the two-layer separation suggests a trust-as-CRDT extension where trust evidence (equivocation proofs, Merkle-root divergence, anomaly scores) propagates through Layer~1 as a monotonic lattice, and a trust-gated merge at Layer~2 rejects contributions once the converged trust score falls below threshold.
Whether such an extension can match the threat coverage of probabilistic robust aggregation in practice is open; integration with existing Byzantine-resilient methods~\cite{blanchard2017machine} is an obvious complement.

\paragraph{Scale-dependent associativity.}
The observation that strategies failing on controlled tensors pass within tolerance on production-scale weights (Section~\ref{sec:cross_scale}) raises theoretical questions about the relationship between algebraic properties and numerical tolerance in high-dimensional spaces.
Analysis of when associativity violations vanish at scale could inform design of ``nearly associative'' strategies.

\section{Per-Strategy Formal Analyses}\label{app:strategy_analysis}

The main text (Section~\ref{sec:problem}) presents representative analyses for weight averaging and SLERP.
We analyse the remaining major strategy families here.

\paragraph{TIES-Merging.}
TIES~\cite{yadav2023ties} operates through trimming, sign election, and disjoint merge.
Sign election over sets is commutative (majority vote does not depend on enumeration order), but the binary merge operation is order-dependent because trimming on pairs versus the full set discards different entries~\cite{yadav2023ties}.
Associativity fails because trimming thresholds depend on the set of vectors being merged---merging $(a,b)$ first applies a different threshold than merging $(b,c)$ first.
Idempotency fails because trimming thresholds are recomputed, and the trim-then-merge pipeline on $\{a,a\}$ need not recover $a$.

\paragraph{DARE.}
DARE~\cite{yu2024dare} applies a stochastic mask $m \sim \mathrm{Bernoulli}(1-p)$ and rescales by $1/(1-p)$.
All three properties fail: the stochastic mask produces different results on each invocation (violating commutativity and idempotency), rescaling factors compound under composition (violating associativity), and the mask differs per call (violating idempotency).
In Phase~1 testing, stochastic strategies were evaluated without fixed seeds to reflect their default behaviour.
The CRDT architecture (Phase~2) resolves this by deriving deterministic seeds from the Merkle root (Section~\ref{sec:layer2}).

\paragraph{Fisher-Weighted Merging.}
Fisher merging is commutative because summation is commutative~\cite{matena2022merging}.
However, associativity fails because intermediate Fisher information is lost during pairwise merging: the Fisher matrix of a merged model is not the sum of the constituent Fisher matrices.
Idempotency holds: merging a model with itself using identical Fisher weights returns the original model.

The remaining strategies (21 of 26) follow similar patterns: stochastic strategies (evolutionary, genetic merge) fail all three properties; sparsification methods (model breadcrumbs, split--unlearn merge) fail idempotency; and all strategies involving normalisation or nonlinear composition fail associativity.
Complete per-strategy results are in Tables~\ref{tab:phase1_controlled} and~\ref{tab:phase1_production}.
\endgroup

\section{Data Flow Example}\label{app:data_flow}

We illustrate the data flow with a two-node merge scenario.
Consider nodes $N_1$ and $N_2$, each with initial states $S_1$ and $S_2$:

\begin{enumerate}
\item $N_1$ fine-tunes a base model and calls $\mathrm{add}(S_1, \theta_1)$, producing state $S'_1$ with updated add set, version vector, and Merkle hash.
\item $N_2$ independently fine-tunes and calls $\mathrm{add}(S_2, \theta_2)$, producing $S'_2$.
\item When $N_1$ and $N_2$ synchronise (in either order), both compute $\mathrm{merge}(S'_1, S'_2) = \mathrm{merge}(S'_2, S'_1)$ by commutativity~\cite{shapiro2011conflict}.
\item Both nodes now have identical visible sets: $\{\theta_1, \theta_2\}$.
\item Both nodes call $\mathrm{resolve}(\cdot, \sigma, \cdot)$, sorting by hash, seeding randomness identically, and obtaining the same merged model $\theta^*$.
\end{enumerate}

For multi-party convergence with $k > 2$ nodes, associativity guarantees that the order of pairwise state merges does not affect the final state~\cite{shapiro2011comprehensive}.
Whether node $N_3$ merges first with $N_1$ or $N_2$, the final visible set---and therefore the resolved model---is identical once all states have been exchanged.

\section{Model Details}\label{app:models}

Table~\ref{tab:models} lists the HuggingFace identifiers for all models used in Tier~2 experiments.

\begin{table*}[t]
\centering
\caption{HuggingFace model identifiers for Tier~2 evaluation.}
\label{tab:models}
\footnotesize
\begin{tabular}{lp{0.65\textwidth}}
\toprule
Role & Identifier \\
\midrule
\multicolumn{2}{l}{\textbf{GPT-2-XL (1.5\,B parameters)}} \\
Base & \texttt{openai-community/\allowbreak gpt2-xl} \\
Instruct & \texttt{nicholasKluge/\allowbreak Aira-2-1B5} \\
Domain & \texttt{lgaalves/\allowbreak gpt-2-xl\_camel-ai-physics} \\
Wiki & \texttt{Clover-Hill/\allowbreak gpt2-xl-finetuned-wikitext103} \\
\midrule
\multicolumn{2}{l}{\textbf{Mistral-7B-v0.1 (7.24\,B parameters)}} \\
Base & \texttt{mistralai/\allowbreak Mistral-7B-v0.1} \\
Instruct & \texttt{mistralai/\allowbreak Mistral-7B-Instruct-v0.2} \\
Hermes & \texttt{NousResearch/\allowbreak Nous-Hermes-2-Mistral-7B-DPO} \\
Zephyr & \texttt{HuggingFaceH4/\allowbreak zephyr-7b-beta} \\
\bottomrule
\end{tabular}
\end{table*}

\section{Multi-Node Convergence Suite Results}\label{app:convergence}

All experiments use the \texttt{crdt-merge} library v0.9.4.

\subsection{Multi-Node Convergence}

Table~\ref{tab:convergence} reports convergence results for 100 nodes across 20 independently randomised gossip orderings.
Every ordering produces a bitwise-identical resolved model, confirming that the \CMS{} architecture achieves strong eventual consistency regardless of communication order.

\begin{table}[htbp]
\centering
\caption{100-node convergence across 20 random gossip orderings.
Strategy: \texttt{slerp}; tensor: $512\times 512$ ($262{,}144$ params per contribution); merges per ordering: 9{,}900.}
\label{tab:convergence}
\footnotesize
\begin{tabular}{rcccl}
\toprule
Ordering & Gossip & Resolve & Max Diff & Status \\
\midrule
1 & 503.1\,ms & 21{,}764\,ms & 0 & PASS \\
2 & 523.6\,ms & 20{,}243\,ms & 0 & PASS \\
3 & 457.4\,ms & 19{,}728\,ms & 0 & PASS \\
4 & 446.7\,ms & 19{,}137\,ms & 0 & PASS \\
5 & 472.1\,ms & 19{,}670\,ms & 0 & PASS \\
6 & 482.7\,ms & 19{,}611\,ms & 0 & PASS \\
7 & 541.7\,ms & 19{,}553\,ms & 0 & PASS \\
8 & 456.4\,ms & 20{,}133\,ms & 0 & PASS \\
9 & 465.5\,ms & 19{,}629\,ms & 0 & PASS \\
10 & 445.1\,ms & 19{,}885\,ms & 0 & PASS \\
11 & 459.5\,ms & 19{,}912\,ms & 0 & PASS \\
12 & 479.0\,ms & 18{,}830\,ms & 0 & PASS \\
13 & 624.7\,ms & 18{,}316\,ms & 0 & PASS \\
14 & 430.5\,ms & 20{,}551\,ms & 0 & PASS \\
15 & 481.5\,ms & 19{,}855\,ms & 0 & PASS \\
16 & 612.1\,ms & 20{,}569\,ms & 0 & PASS \\
17 & 486.4\,ms & 20{,}642\,ms & 0 & PASS \\
18 & 562.9\,ms & 18{,}621\,ms & 0 & PASS \\
19 & 460.8\,ms & 18{,}822\,ms & 0 & PASS \\
20 & 464.4\,ms & 18{,}528\,ms & 0 & PASS \\
\midrule
\multicolumn{2}{l}{Avg gossip: 492.8\,ms} & \multicolumn{3}{l}{All orderings bitwise equal: \textbf{YES}} \\
\bottomrule
\end{tabular}
\end{table}

\subsection{Network Partition and Healing}

One hundred nodes are split into 10 partitions (10 nodes each).
Each partition gossips internally and converges to a distinct, consistent hash.
After partition healing, all 100 nodes converge to a single bitwise-identical result.
The final hash matches the multi-node convergence result, confirming deterministic SEC recovery.

\begin{table}[htbp]
\centering
\caption{Network partition and healing: 100 nodes split into 10 isolated partitions, then healed.
Each partition converges to a distinct hash during isolation; after healing, all nodes converge to a single bitwise-identical result matching the unpartitioned experiment (Table~\ref{tab:convergence}).}
\label{tab:partition}
\footnotesize
\begin{tabular}{lr}
\toprule
Metric & Value \\
\midrule
Nodes / Partitions & 100 / 10 \\
Partition gossip time & 4.6\,ms \\
Distinct partition hashes & 10/10 \\
Healing time & 556.6\,ms \\
Post-healing convergence & 100/100 nodes \\
Bitwise identical & YES \\
Final hash matches Table~\ref{tab:convergence} & YES \\
\bottomrule
\end{tabular}
\end{table}

\subsection{Cross-Strategy Convergence Sweep}

Table~\ref{tab:strategy_sweep} verifies that all 26 merge strategies converge to a single canonical hash across 10 nodes.
This confirms that the two-layer architecture provides strategy-independent convergence: every strategy, regardless of its algebraic properties, produces an identical resolved model on every node.
Note that population-based strategies (\texttt{evolutionary\_merge}, \texttt{genetic\_merge}) incur substantially higher resolve times due to their internal search processes.

\begin{table}[htbp]
\centering
\caption{Cross-strategy convergence: all 26 strategies on 10 nodes, $64\times 64$ tensors.
All strategies produce the same canonical hash, confirming strategy-independent convergence.
Note: \texttt{evolutionary\_merge} and \texttt{genetic\_merge} exhibit resolve times approaching $90$\,s due to population-based search; all other strategies resolve in under $200$\,ms.}
\label{tab:strategy_sweep}
\footnotesize
\begin{tabular}{lrrl}
\toprule
Strategy & Gossip & Resolve & Status \\
\midrule
\texttt{ada\_merging} & 0.6\,ms & 200\,ms & PASS \\
\texttt{adarank} & 0.6\,ms & 102\,ms & PASS \\
\texttt{dam} & 0.7\,ms & 112\,ms & PASS \\
\texttt{dare} & 0.6\,ms & 80\,ms & PASS \\
\texttt{dare\_ties} & 1.0\,ms & 79\,ms & PASS \\
\texttt{della} & 0.8\,ms & 122\,ms & PASS \\
\texttt{dual\_projection} & 0.6\,ms & 26\,ms & PASS \\
\texttt{emr} & 0.6\,ms & 22\,ms & PASS \\
\texttt{evolutionary\_merge} & 0.7\,ms & 86{,}400\,ms & PASS \\
\texttt{fisher\_merge} & 0.6\,ms & 7\,ms & PASS \\
\texttt{genetic\_merge} & 0.5\,ms & 87{,}960\,ms & PASS \\
\texttt{led\_merge} & 0.6\,ms & 143\,ms & PASS \\
\texttt{linear} & 0.7\,ms & 11\,ms & PASS \\
\texttt{model\_breadcrumbs} & 0.6\,ms & 22\,ms & PASS \\
\texttt{negative\_merge} & 0.7\,ms & 8\,ms & PASS \\
\texttt{regression\_mean} & 0.7\,ms & 8\,ms & PASS \\
\texttt{representation\_surgery} & 0.7\,ms & 9\,ms & PASS \\
\texttt{safe\_merge} & 0.6\,ms & 29\,ms & PASS \\
\texttt{slerp} & 0.6\,ms & 14\,ms & PASS \\
\texttt{split\_unlearn\_merge} & 0.6\,ms & 43\,ms & PASS \\
\texttt{star} & 0.8\,ms & 101\,ms & PASS \\
\texttt{svd\_knot\_tying} & 0.8\,ms & 154\,ms & PASS \\
\texttt{task\_arithmetic} & 0.7\,ms & 8\,ms & PASS \\
\texttt{ties} & 0.6\,ms & 29\,ms & PASS \\
\texttt{weight\_average} & 0.7\,ms & 8\,ms & PASS \\
\texttt{wt\_scope\_alignment} & 0.7\,ms & 18\,ms & PASS \\
\bottomrule
\end{tabular}
\end{table}

\subsection{Scalability Benchmark}

Table~\ref{tab:scalability} measures how gossip and resolve times scale as the number of participating nodes increases from 2 to 50.
Gossip time grows quadratically in the number of nodes (reflecting all-pairs state exchange), while per-call \texttt{merge()} cost remains constant in tensor size.
As noted in Section~\ref{sec:convergence_suite}, this prototype gossip protocol is designed for validation purposes; production deployments beyond ${\sim}50$ nodes would benefit from optimised dissemination protocols.

\begin{table}[htbp]
\centering
\caption{Scalability: \texttt{slerp} on $64 \times 64$ tensors, 2--50 nodes.
Gossip time scales as $O(n^2)$ (all-pairs merges); per-call \texttt{merge()} is $O(1)$ in tensor size.}
\label{tab:scalability}
\footnotesize
\begin{tabular}{rrrr@{\hskip 6pt}r@{\hskip 6pt}l}
\toprule
Nodes & Params & Merges & Gossip & Resolve & Status \\
\midrule
2 & 8{,}192 & 2 & 0.0\,ms & 0.8\,ms & PASS \\
5 & 20{,}480 & 20 & 0.1\,ms & 3.7\,ms & PASS \\
10 & 40{,}960 & 90 & 0.6\,ms & 16.8\,ms & PASS \\
20 & 81{,}920 & 380 & 4.2\,ms & 74.7\,ms & PASS \\
30 & 122{,}880 & 870 & 12.3\,ms & 163.9\,ms & PASS \\
50 & 204{,}800 & 2{,}450 & 127.3\,ms & 445.2\,ms & PASS \\
\bottomrule
\end{tabular}
\end{table}

\end{document}